\documentclass[11pt,twoside,letterpaper]{llncs}
\usepackage{times}
\usepackage[T1]{fontenc}   
\usepackage{amsmath}
\usepackage{amssymb}
\usepackage{fancybox}
\usepackage{graphics,graphicx}
\newtheorem{fact}{Fact}
\newcommand{\junk}[1]{}

\junk{
 \documentclass[runningheads]{llncs}

\usepackage{amsmath,amssymb}
\usepackage{graphicx}
\usepackage{fancybox}
\usepackage{algorithmic}

\newcommand{\junk}[1]{}

\newtheorem{fact}{Fact}

\let\tilde\widetilde

}

\begin{document}
\mainmatter

\title{Consequences of APSP, triangle detection, and 3SUM  hardness
  for separation between determinism and non-determinism\thanks{Research
supported in part by VR grant 2017-03750 (Swedish Research Council).}}


\author{
  Andrzej Lingas\inst{1}
}

\institute{
  Department of Computer Science, Lund University, 22100 Lund, Sweden.
  \newline
  E-mail:~\texttt{Andrzej.Lingas@cs.lth.se}
}

\date{}


\maketitle
\begin{abstract}
Let $NDTIME(f(n),g(n))$ denote the class of problems solvable in
  $O(g(n))$ time by a multi-tape Turing machine using an $f(n)$-bit
  non-deterministic oracle, and let $DTIME(g(n))=NDTIME(0,g(n)).$ We
  show that if the all-pairs shortest paths problem (APSP) for
  directed graphs with $N$ vertices and integer edge weights within a
  super-exponential range $\{ -2^{N^{k+o(1)}},...,2^{N^{k+o(1)}}\}$,
  $k\ge 1,$ does not admit a truly subcubic algorithm then for any
  $\epsilon > 0,$
  \newline
  $NDTIME(\lceil \frac 12 \log_2 n \rceil ,
  n)\nsubseteq DTIME(n^{1+\frac 1{2+k}-\epsilon}).$ If the APSP
  problem does not admit a truly subcubic algorithm already when the
  edge weights are of moderate size then we obtain even a stronger
  implication, namely that for any $\epsilon > 0,$ $NDTIME(\lceil
  \frac 12 \log_2 n\rceil , n)\nsubseteq DTIME(n^{1.5-\epsilon}).$
\junk{On the other hand,
we prove that one may assume w.l.o.g.
that the integer edge weights are at most in
$\{ -2^{N^{2+o(1)}},...,2^{N^{2+o(1)}}\}$,
where $N$ is the number of vertices in the input graph.}
Similarly, we show that if the triangle detection problem (DT) in a
graph on $N$ vertices does not admit a truly sub-$N^{\omega}$-time
algorithm then for any $\epsilon > 0,$ $NDTIME(\lceil \frac 12\log_2 n
\rceil ,n)\nsubseteq DTIME(n^{\omega/2-\epsilon}),$ where $\omega$
stands for the exponent of fast matrix multiplication.
\junk{
Also, we present a stronger
  implication from a conjecture on the more general problem of
  detecting a maximum weight $\ell$-clique.
    We present also a
  stronger implication from the conjecture that the
  more general problem of
  detecting a maximum weight $\ell$-clique does not admit a truly
  sub-$n^{\ell}$-time algorithm.}
For the more general problem of detecting a minimum weight
$\ell$-clique ($MWC_{\ell}$) in a graph with edge weights of moderate
size, we show that the non-existence of truly sub-$n^{\ell}$-time
algorithm yields for any $\epsilon >0$, $NDTIME((\ell -2)\lceil \frac
12 \log_2 n \rceil,\ n)\nsubseteq DTIME(n^{1+\frac {\ell
    -2}{2}-\epsilon})$.  Next, we show that if 3SUM for $N$ integers
in $\{-2^{N^{k+o(1)}},...,2^{N^{k+o(1)}}\}$ for some $k\ge 0,$ does
not admit a truly subquadratic algorithm then for any $\epsilon >0,$
$NDTIME(\lceil \log_2 n \rceil ,\ n)\nsubseteq$ $DTIME(n^{1+\frac
  1{1+k}-\epsilon }).$ Finally, we observe that the Exponential Time
Hypothesis (ETH) implies $NDTIME(\lceil k\log_2 n \rceil, n)
\nsubseteq DTIME(n)$ for some $k>0,$ while the strong ETH (SETH)
yields for any $\epsilon>0,$ $NDTIME(\lceil \log_2 n \rceil,
n)\nsubseteq DTIME(n^{2-\epsilon})$.  For comparison, the strongest
known result on separation between non-deterministic and deterministic
time only asserts
\newline
$NDTIME(O(n), n)\nsubseteq DTIME(n).$
\end{abstract}
\section{Introduction}\label{intro}
The complexity status of the all-pairs shortest paths problem (APSP)
in directed graphs with arbitrary edge weights is a major open
problem in the area of graph algorithms. In spite of several decades
of research efforts, no truly subcubic 
algorithm for this
problem is known. 

By truly subcubic, Vassilevska Williams and Williams \cite{VWW} mean
\newline
$O(N^{3-\delta}poly(\log M))$ for some $\delta >0$, where $N$ is the
number of vertices in the input graph or the number of rows and
columns in the input matrix, and the edge weights or matrix entries
respectively are in the range $\{ -M,...,M\}$. This definition
assumes that $M$ is not too large.  When $M$ is very
large, e.g., $M=2^{N^{\phi}}$ for some $\phi > 0,$ then $poly(\log M)$
can become at least polynomial in $N.$ For this reason, we shall adopt a more
strict definition of truly subcubic, namely $O(N^{3-\delta}(\log
M)^{1+o(1)})$ for some $\delta >0$.  This definition is still
compatible with the reductions presented in \cite{VWW}, in particular,
it allows for multiplication of $O(\log M)$ bit numbers in $(\log
M)^{1+o(1)}$ time, and it works for $M$ of $2^{N^{k+o(1)}}$ size,
for any constant $k\ge 0.$

Vassilevska Williams and Williams presented a list
of eleven problems that they could show to be equivalent to the APSP
problem regarding the question of admitting a truly
subcubic algorithm \cite{VWW}. 
Thus, if any problem on the list
could be shown to admit a truly subcubic algorithm then
all the remaining problems on the list, in particular APSP, would have
truly subcubic algorithms.
Besides APSP and the verification
of the  naturally related distance (i.e., $(\min,+)$) 
matrix product, the list includes  problems of different form ranging from
multi-functions to decision problems. 
Some of the problems on the list immediately reduce to search problems.
In particular,
the problem of detecting a triangle of negative total edge weight (DNT)
belongs to the latter ones. For these reasons, the problems on the list
admit quite different upper time bounds in the quantum computational
model (cf. Table 2 in Appendix) or in the non-deterministic Turing machine or RAM model.
We utilize the presence of the problems on the list that directly reduce
to search problems in order to derive among other things
 the following
 implication, where
$NDTIME(f(n),g(n))$ denotes the class of problems solvable in
$O(g(n))$ time by a multi-tape Turing machine using an $f(n)$-bit
non-deterministic oracle, and $DTIME(g(n))=NDTIME(0,g(n)).$
If the APSP problem for directed graphs
with integer edge weights of moderate size
 does not admit a truly
subcubic algorithm then for any $\epsilon > 0,$
$NDTIME(\lceil \frac 12 \log_2 n \rceil, n)\nsubseteq DTIME(n^{1.5-\epsilon}).$
By an edge weight of moderate size, we mean a weight requiring $N^{o(1)}$-bit representation, where $N$ is the number of vertices in the
input graph or the number of rows and columns
in the input matrix, respectively. Observe that the best known and celebrated result
separating non-deterministic time from the deterministic one
in terms of our notation is just
$NDTIME(O(n), n)\nsubseteq DTIME(n)$ \cite{PPST}
(cf. \cite{Aj}).
Furthermore, no result of the form $NDTIME(O(n^q), n^q)\nsubseteq DTIME(n^q)$
for $q> 1$ is known.
Our more general result is as follows. 
\junk{
 If the APSP problem or any of the remaining problems
on the list with integer edge weights or integer matrix entries
does not admit a truly subcubic
algorithm then for any $\epsilon' >0,$
$NDTIME(\frac 12\log_2 n,\ n)\nsubseteq DTIME(n^{1.25{2+k}-\epsilon'}).$}
\par
\noindent
{\em Let $k\ge 0,$
and let  $Q$ stand for the set of integers $\{ -2^{N^{k+o(1)}},...,2^{N^{k+o(1)}}\}.$
If for any $\epsilon > 0,$
the APSP problem or any of the subcubic-time equivalent
problems for graphs on $N$ vertices  with edge weights in $Q$
or $N\times N$ matrices with entries in $Q$ does not
admit an $O(N^{3-\epsilon + k+o(1)})$-time algorithm then
for any $\epsilon' >0,$
$NDTIME(\lceil \frac 1{2+k}
\log_2 n \rceil,\ n)\nsubseteq DTIME(n^{1+\frac 1{2+k}-\epsilon'})$ holds.}
\par
\noindent
Thus, if it is true that the APSP problem for directed graphs
does not admit a truly
subcubic algorithm when the edge weights are
within a super-exponential range
then showing this seems beyond the
reach of presently known techniques. Simply, it would imply
an enormous breakthrough not only in
lower bounds on Boolean circuit size for
natural problems but also in separation between deterministic
and non-deterministic time.
\junk{
On the other hand,
we show that one may assume w.l.o.g.
that the integer edge weights and matrix entries for the problems
on the list are at most in
$R=\{ -2^{N^{2+o(1)}},...,2^{N^{2+o(1)}}\}$,
where $N$ is the number of vertices in the input graph
or the number of rows and columns in the input matrix,
respectively. Therefore, for ranges substantially larger
than $R$ these problems admit truly subcubic algorithms
(according to our definition)
by the reduction to the range $R$  and an application
of a straightforward cubic algorithm
within this range.}

We also consider the much simpler problem of detecting a triangle (DT)
in a (undirected) graph. It immediately reduces to a search problem
and in dense graphs it can be solved in $O(n^{\omega})$ time by a well
known reduction to fast matrix multiplication \cite{IR}.  The somewhat
informal concept of a truly subcubic algorithm can be naturally
generalized to include that of a truly sub-$N^{\delta}$-time algorithm
by replacing $N^{3-\epsilon}$ with $N^{\delta-\epsilon},$ respectively
\cite{VWW}.  Similarly, we obtain the following weaker implication: if
the DT problem for a graph on $N$ vertices does not admit a truly
sub-$N^{\omega}$-time algorithm then for any $\epsilon > 0,$
$NDTIME(\lceil \frac 12\log_2 n \rceil,n)\nsubseteq
DTIME(n^{\omega/2-\epsilon}).$

Next, we consider the problem of detecting a triangle of minimum total
weight (MWT) in an edge weighted graph.  Note that the DNT and DT
problems can be easily reduced to the MWT problem. It follows in
particular that if the MWT problem admits a truly subcubic algorithm
then any problem on the aforementioned list has also this property
(cf. Conclusion in \cite{VWW}). A natural generalization of the MWT
problem is that of determining a minimum weight $\ell$-clique in an
edge weighted graph (MWC-$\ell$) conjectured to not admit a truly
sub-$N^{\ell }$-time algorithm \cite{ABD}. We show that if this
conjecture holds for graphs with $N$ vertices and edge weights in $\{
-2^{N^{k+o(1)}},...,2^{N^{k+o(1)}}\}$ then for any positive $\epsilon
>0$, $NDTIME((\ell -2)\lceil \frac 1{2+k} \log_2 n
\rceil,\ n)\nsubseteq $ $DTIME(n^{1+\frac {\ell -2}{2+k}-\epsilon})$

The 3SUM problem, which is to decide if an
input set of numbers contains three numbers summing to zero, is widely
believed to not admit a truly subquadratic algorithm (weakened 3SUM
hypothesis). For this reason, one has shown truly subquadratic
reducibility of 3SUM to several other problems believed to have
almost quadratic time complexity in order to demonstrate their relative
hardness (see, e.g., \cite{DKP,LPV,VWW}).  We show that if 3SUM, when the
input $N$ numbers are integers in $\{
-2^{N^{k+o(1)}},...,2^{N^{k+o(1)}}\}$ for some $k\ge 0,$ does not
admit a truly subquadratic algorithm then for any $\epsilon >0,$ 
\newline
$NDTIME(\lceil \frac 1{1+k} \log_2 n \rceil,\ n)\nsubseteq DTIME(n^{1+\frac 1{1+k}-\epsilon })$ holds.

Finally, we observe that the Exponential Time Hypothesis (ETH)
\cite{IP} implies $NDTIME(\lceil k\log_2 n \rceil, n)
\nsubseteq DTIME(n)$ for some $k>0$, while the strong ETH 
(SETH) \cite{CIP} yields for any
$\epsilon>0,$ $NDTIME(\lceil \log_2 n \rceil, n)\nsubseteq
DTIME(n^{2-\epsilon})$.

Our implications from the known conjectures are in a form of a negated
containment of a linear-time with a logarithmic non-deterministic oracle
in a respective deterministic bounded-time class. For interesting or
even dramatic consequences of this kind of containment
see subsection 1.1.2 in \cite{Ryan}.

Marginally, 
we present simple $\tilde{O}(N^{1.5})$ time quantum algorithms for the MWT problem
and the problem of
verifying if an $N\times N$ matrix defines a metric (MDM),
occurring on the list in \cite{VWW}.
\junk{and running in
$\tilde{O}(N^{1.5})$ time provided that an oracle for the weighted
adjacency matrix representing an edge weighted graph on $N$ vertices
is given.
The notation $\tilde{O}(\ )$ suppresses poly-logarithmic in
$N$ factors.  Observe that for the simpler DT problem even faster
quantum algorithms are known \cite{MS} while on the other hand the
best known quantum algorithms for APSP are substantially slower
\cite{NV}, see Table 2 (Appendix). We also present a similar
$\tilde{O}(N^{1.5})$-time quantum algorithm for the problem of
verifying if an $N\times N$ matrix defines a metric (MDM), occurring
on the list in \cite{VWW}.}

\begin{remark}
  All our main results
  can be
viewed as putting straightforward/obvious and/or known implications of
the conjectures in a common framework. This in particular shows subtle
differences between the implications and exhibits dependency on the
range parameter.  Thus, the contribution of our paper is mostly a
conceptual one formalizing a known intuition.
\end{remark}

\begin{table*}[t]
\begin{center}
\begin{tabular}{||c|c||} \hline \hline
conjecture on & implication
\\ \hline \hline
APSP and equiv. &  for $\epsilon > 0,$
$NDTIME(\lceil \frac 12 \log_2 n\rceil , n)\nsubseteq DTIME(n^{1.5-\epsilon})$
\\ \hline
DT  & for $\delta < \omega$,
$NDTIME(\lceil \frac 12\log_2 n \rceil ,\ n)\nsubseteq DTIME(n^{\delta/2})$
\\ \hline
DC-$\ell$ where $\ell|3$ & for $\epsilon>0$,
$NDTIME((\ell -2)\lceil \frac 12 \log_2 n \rceil ,\ n)\nsubseteq DTIME(n^{\omega  \ell /6-\epsilon})$
	\\ \hline
MWC-$\ell$ &  for $\epsilon
>0$, $NDTIME((\ell -2)\lceil \frac 12 \log_2 n
\rceil,\ n)\nsubseteq DTIME(n^{1+\frac {\ell -2}{2}-\epsilon})$ 
\\ \hline
3SUM & for $\epsilon >0$, $NDTIME(\lceil \log_2 n \rceil, n) \nsubseteq DTIME(n^{2-\epsilon})$
\\ \hline
ETH &  for some $k>0$, $NDTIME(\lceil k\log_2 n \rceil, n) \nsubseteq DTIME(n)$ 
\\ \hline
SETH & for $\epsilon>0,$
  $NDTIME(\lceil \log_2 n \rceil, n)\nsubseteq
  DTIME(n^{2-\epsilon})$
\\ \hline \hline
\end{tabular}
\label{table: 1}
\vskip 0.5cm
\caption{Implications from the conjectures when the input edge
  weights or input matrix entries or input
  numbers are assumed to be integers in $\{ -2^{N^{o(1)}},\ldots,2^{N^{o(1)}}   \}$.
For the definition of DC-$\ell$, see Section 3.}
\end{center}
\end{table*}

\section{Preliminaries}
For a positive integer $r,$
$[r]$ will denote $\{1,...,r\}$.

Among the eleven problems on the list of subcubic-time equivalent
problems in \cite{VWW}, we shall refer explicitly to:
\begin{itemize}
\item
\noindent
the all-pairs shortest path problem in
directed edge weighted graphs (APSP), 
\item
the problem of detecting a triangle of 
negative total edge weight (DNT), and
\item
the problem of verifying if a matrix
defines a metric (MDM).
\end{itemize}
Note that an $N\times N$ matrix $K$ 
defines a metric on $[N]$ if and only if
it has non-negative entries, $K[i,j]=0$
iff $i=j$ and
$K[i,j]=
K[j,i]$ for all $i,\ j\in [N],$ and  $K[i,j]\le K[i,k]+K[k,j]$ for all
$i,\ j. \ k \in [N].$
Of course, the first three conditions can be easily verified
in quadratic time.
\junk{ and perhaps therefore, such a matrix $K$
is said to define a metric in \cite{VWW} if it
just fulfills the last, i.e., the triangle inequality
condition.}
\par
We shall also consider the problem of finding a 
triangle of minimum total edge weight in an
edge weighted graph (MWT) and the simpler problem
of detecting a triangle in a (unweighted) graph (DT),
as well as their generalizations where a triangle
is replaced by an $\ell$-clique (MWC-$\ell$ and
DC-$\ell$, respectively).
Note that DNT and DT trivially reduce to MWT.
\par
Furthermore, we shall consider the 3SUM problem
which is to decide if a given set of numbers
contains three elements whose sum is zero,
and the Exponential Time Hypothesis (ETH) \cite{IP}
as well as its strong version (SETH) \cite{CIP}.
\par
The reductions proving the subcubic-time equivalences
between the eleven problems on the list in Theorem 1.1 in \cite{VWW}
do not introduce new very large edge  weights or matrix
entries. They typically use the edge weights or
matrix entities from the reduced
problem. The exception is the use of $+\infty $ or $-\infty$
in case of some problems on the list. However, the latter
can be easily simulated by multiplying the maximum or
minimum of the assumed range by a polynomial in the 
number of vertices or in the number of matrix rows/columns
in the considered
problem.

In Definition 3.1 in \cite{VWW}, one requires an
$O(m^{3-\delta})$ bound on the time taken by a subcubic reduction,
where $m=N\log M$ in our terms. This definition as that of truly subcubic
also assumes that the edge weights or matrix entries are not too large.
We can replace  the required upper time-bound by
\newline
$O(N^{3-\delta}(\log M)^{1+o(1)})$
to extend the edge weight or matrix entry range to at least
$\{ 2^{N^{k+o(1)}},...,2^{N^{k+o(1)}}\}$ for any fixed $k\ge 0.$
In fact, the authors show in Section 4.3  of \cite{VWW} that
if random bits are allowed then the
polylogarithmic dependence 
on $M$ can be replaced
by a polylogarithmic dependence on $N$
in the aforementioned reductions.

Thus, Theorem 1.1 in \cite{VWW} holds also when the edge weights
or matrix entries in the problems on the list are in the
range $\{-2^{-N^{k+o(1)}},...,2^{-N^{k+o(1)}}\}$ for some $k\ge 0$
(under the assumption of the more strict definition of
``truly subcubic'' from the introduction).
Hence, we have the following fact.

\begin{fact}\label{fact: VWW}
Let $k\ge 0.$
If the APSP problem or the DNT problem, or the problem
of verifying if a matrix defines a metric, or 
any of the remaining problems on the list in \cite{VWW},
does not admit a truly subcubic
algorithm when the edge weights or matrix entries
are in  $\{ -2^{N^{k+o(1)}},...,2^{N^{k+o(1)}}\}$
 then none of the problems admits a truly
subcubic algorithm
when the edge weights or 
matrix entries are in $\{ -2^{N^{k+o(1)}},...,2^{N^{k+o(1)}}\}.$
\end{fact} 

\junk{
  Finally, we remark that the definition of a subcubic reduction
between problems introduced in \cite{VWW} seems to suffer
a similar deficiency as that of a truly subcubic algorithm, when
the range $\{-M,...,M\}$ of edge weights or matrix entries becomes very large,
e.g., exponential in the number $N$ of vertices of the input graph
or the number $N$ of rows/columns of the input matrix.
In Definition 3.1 in \cite{VWW}, one requires an
$O(m^{3-\delta})$ bound on the time taken by a subcubic reduction,
where $m=N\log M$ in our terms. This definition as that of truly subcubic
also assumes that the edge weights or matrix entries are not too large.
We can replace  the required upper time-bound by
\newline
$O(N^{3-\delta}(\log M)^{1+o(1)})$
to extend the edge weight or matrix entry range to at least
\newline
$\{ 2^{N^{k+o(1)}},...,2^{N^{k+o(1)}}\}$ for any fixed $k\ge 0.$
In fact, the authors show in Section 4.3  of \cite{VWW} that
if random bits are allowed then the
polylogarithmic dependence 
on $M$ can be replaced
by a polylogarithmic dependence on $N$
in the aforementioned reductions.}

\section{Implications of APSP and DT hardness}

To start, we observe that the decision version of the MWT
problem, in particular of the DNT problem, as well as the complement of
the problem of verifying if a matrix defines a metric admit
$N^{2+k+o(1)}$-time algorithms with a non-deterministic
$\lceil \log_2 N\rceil$-bit
oracle, when the edge weights or matrix entries are in the set $\{
2^{N^{k+o(1)}},...,2^{N^{k+o(1)}}\}$ for some $k\ge 0.$
Recall that $N$ stands
for the number of vertices in the input graph or the number of
rows/columns in the input matrix, respectively.

\begin{lemma}\label{lem: guess}
Let $k\ge 0,$ let 
 $Q$ denote the set of integers $\{ -2^{N^{k+o(1)}},...,2^{N^{k+o(1)}}\},$
and let $d\in Q.$
The problem of determining if a graph with 
$N$ vertices and edge weights in $Q$
has a triangle of 
weight smaller than $d,$ as well as the problem of verifying
that an $N\times N$ matrix with entries in $Q$
 does not define a metric are in 
$NDTIME(\lceil \log_2 N\rceil , N^{2+k+o(1)})$.
\end{lemma}
\begin{proof}
  In order to guess a vertex of a triangle of edge weight smaller than
  $d$, a non-deterministic $\lceil \log_2 N \rceil$-bit oracle is
  sufficient.  A multi-tape Turing machine can easily verify if the
  guessed vertex belongs to a triangle of edge weight smaller than
  $d,$ and if so return such a triangle. Simply, for each pair of
  vertices it can examine if the pair jointly with the guessed vertex
  forms a triangle of total edge weight smaller than $d,$ and if so
  output the triangle in $N^{2+k+o(1)}$ total time.

A deterministic multi-tape Turing machine can also easily verify if an
input $N\times N$ matrix $K$ satisfies the first three condition
required by a metric, including the symmetry one, in $N^{2+k+o(1)}$
time. If the three conditions are satisfied it remains to check if
$K[i,j]>K[i,k]+K[k,j]$ for some $i,\ j. \ k.$ Again, to guess the
first index belonging to such a triple of indices a non-deterministic
$\lceil \log_2 N \rceil$-bit oracle is sufficient.  Similarly, to
verify if the guessed first index, say $i,$ belongs to a triple of
indices violating the triangle inequality condition can be done by
checking if $K[i,j]>K[i,k]+K[k,j]$ for all other possible indices
$j,\ k.$ Again, it can be easily done by multi-tape Turing machine in
$N^{2+k+o(1)}$ time.
\qed
\end{proof}

\begin{theorem}\label{theo: APSP}
Let $k\ge 0,$
and let  $Q$ stand for the set of integers 
\newline
$\{ -2^{N^{k+o(1)}},...,2^{N^{k+o(1)}}\}.$
If for any $\epsilon > 0,$
the APSP problem or any of the subcubic-time equivalent
problems for graphs on $N$ vertices  with edge weights in $Q$
or $N\times N$ matrices with entries in $Q$ does not
admit an $O(N^{3-\epsilon + k+o(1)})$-time algorithm then
for any $\epsilon' >0,$
$NDTIME(\lceil \frac 1{2+k}\log_2 n \rceil,\ n)\nsubseteq $
\newline
$DTIME(n^{1+\frac 1{2+k}-\epsilon'})$ holds.
\end{theorem}
\begin{proof}
\junk{The reductions proving the subcubic-time equivalences
between the eleven problems on the list in Theorem 1.1 in \cite{VWW}
do not introduce new very large edge  weights or matrix
entries, they typically use the edge weights or
matrix entries entities from the reduced
problem. The exception is the use of $+\infty $ or $-\infty$
in case of some problems on the list. However, the latter
can be easily simulated by multiplying the maximum or
minimum of the assumed range by a polynomial in the 
number of vertices or in the size of matrix of the considered
problem. Thus, Theorem 1.1 holds also when the edge weights
or matrix entries in the problems on the list are in the
range $\{-2^{-N^{k+o(1)}},...,2^{-N^{k+o(1)}}\}$.}
By the theorem assumptions  and Fact \ref{fact: VWW}, we infer
that for any $\epsilon >0,$ 
the problem of detecting a negative triangle (DNT), 
when the edge weights are in the range $\{ -2^{N^{k+o(1)}},...,2^{N^{k+o(1)}}\},$
does not admit an $O(N^{3-\epsilon +k +o(1)})$-time
RAM algorithm under the logarithmic cost.
\junk{ On the other hand, in order to guess such a triangle
a non-deterministic $3\lceil \log_2 N \rceil$-bit oracle is sufficient.
A multi-tape Turing machine can easily compute if the triangle is negative
in $poly(n^{o(1)}$ time (one multiplication?).}
On the other hand,
the so restricted
DNT problem is in $NDTIME(\lceil \log_2 N\rceil , N^{2+k+o(1)})$
by Lemma \ref{lem: guess}. 
Consequently, if we assume $N^{k+o(1)}$-bit representation of
the edge weights and
set $n=N^{2+k+o(1)}$ then we conclude that 
the restricted DNT is
in
\newline
$NDTIME(\lceil \frac 1{2+k}\log_2 n\rceil , n).$  Now the proof is by contradiction.
Suppose that 
$NDTIME(\lceil \frac 1{2+k}\log_2 n\rceil , n) \subseteq DTIME(n^{1+\frac 1{2+k}-\epsilon'})$ holds.
Then, the restricted DNT problem admits an 
$O((N^{2+k+o(1)})^{1+\frac 1{2+k}-\epsilon'})$-time
algorithm, i.e., an $O(N^{3+k-(2+k+o(1))\epsilon' +o(1)})$-time algorithm in
the multi-tape (deterministic) Turing machine model.
A multi-tape (deterministic) Turing machine of time complexity
$T(n)\ge n$ can be easily simulated by a RAM with logarithmic cost
running in $O(T(n)\log T(n))$ time
(see, e.g., section 1.7 in \cite{AHU74}).
We infer that the restricted DNT problem admits
an $O(N^{3+k-(2+k+o(1))\epsilon' +o(1)}\log n )$-time, i.e.,
an $O(N^{3+k-(2+k+o(1))\epsilon' +o(1)})$-time algorithm, in the RAM model.
We obtain a contradiction.
\qed
\end{proof}

In particular, when the edge weights or matrix entries are of moderate size,
we obtain the stronger implication of the following form:
for any $\epsilon' >0,$
\newline
$NDTIME(\lceil \frac 12\log_2 n \rceil,\ n)\nsubseteq DTIME(n^{1.5-\epsilon'})$.
In fact,
APSP is assumed to be hard already when
the weights are polynomial \cite{AVY}.
\junk{
\begin{corollary}\label{cor: APSP}
If for any $\epsilon > 0,$
the APSP problem or any of the subcubic-time equivalent
problems, for graphs on $N$ vertices  with edge weights in
$\{ -2^{N^{o(1)}},...,2^{N^{o(1)}}\},$
or $N\times N$ matrices with entries in $\{ -2^{N^{o(1)}},...,2^{N^{o(1)}}\},$ does not
admit an $O(N^{3-\epsilon + o(1)})$-time algorithm then
for any $\epsilon' >0,$
$NDTIME(\lceil \frac 12\log_2 n \rceil,\ n)\nsubseteq DTIME(n^{1.5-\epsilon'})$
holds.
\end{corollary}}

For the simpler DT problem, we obtain similarly
the following implication.

\begin{theorem}\label{theo: DT}
If for any $\delta <  \omega ,$
the problem of detecting a triangle
in a graph  on $N$ vertices does not admit
an $O(N^{\delta+o(1)})$-time algorithm
then for any positive $\delta' < \omega$,
$NDTIME(\lceil \frac 12\log_2 n \rceil ,\ n)\nsubseteq DTIME(n^{\delta'/2})$ holds.
\end{theorem}
\begin{proof}
The proof is similar to that of Theorem \ref{theo: APSP}.
First, we observe that the triangle detection problem
specified in the theorem
is in
\newline
$NDTIME(\lceil \log_2 N\rceil , N^{2}\log N)$
by modifying slightly the proof of Lemma \ref{lem: guess}.
The difference is that in case of the triangle detection
problem we do not have edge weights and we need to operate
only on vertex indices of logarithmic size.
Also, the verification if three vertices form a triangle
is easier than that they form a triangle of 
total edge weight smaller than $d$
in an edge weighted graph.

Next, we proceed along the lines of the proof of Theorem \ref{theo: APSP}.   
Namely, we set $n=N^2\log_2 N$ to conclude that 
the triangle detection problem is
in
\newline
 $NDTIME(\lceil \frac 12\log_2 n \rceil, n).$  
In order to obtain a contradiction
suppose that
\newline
$NDTIME(\lceil \frac 12\log_2 n \rceil,\ n)\subseteq DTIME(n^{\delta'/2})$ holds for some $\delta'< \omega$.
Then, the triangle detection problem admits an $O((N^2\log N)^{\delta' /2})$-time
algorithm, i.e., an $O(N^{\delta'} poly(\log N))$-time algorithm in
the multi-tape (deterministic) Turing machine model.
\junk{
Since a multi-tape (deterministic) Turing machine of time complexity
$T(n)\ge n$ can be easily simulated by a RAM under the logarithmic-cost
model running in $O(T(n)\log T(n))$ (see, e.g., section 1.7 in \cite{AHU74}),} 
Hence, we can infer 
(by the same argument as in the proof of Theorem \ref{theo: APSP} ) that the triangle detection problem admits
an
\newline
$O(N^{\delta' } poly(\log N)\log N )$-time, i.e.,
an $O(N^{\delta'+o(1)})$-time algorithm, in the RAM model.
We obtain a contradiction.
\qed
\end{proof}

The asymptotically fastest algorithm for
the detection of an $\ell $-clique in an $N$ vertex
graph (DC-$\ell$) is by a straightforward reduction to the triangle problem.
In particular, if $\ell $ is divisible by $3$, it runs
in $O(n^{\omega \ell /3})$ time \cite{NP85}. If $\ell$ is
not divisible by $3$ one uses also
fast rectangular multiplication
and the expression is more complicated \cite{EG04}.
One can conjecture
that the aforementioned asymptotic time cannot be
substantially improved.
We can easily generalize Theorem \ref{theo: DT}
to include the
consequences of such a conjecture (the main trick is to guess
$\ell - 2$ vertices of the clique),
the details are left to the reader.

\begin{theorem}\label{theo: KT}
Let $\ell \ge 3$ be divisible by $3.$
If for any $\epsilon >0,$
the problem of detecting an $\ell$-clique
in a graph  on $N$ vertices does not admit
an $O(N^{\omega \ell /3 - \epsilon +o(1)})$-time algorithm,
then for any positive $\epsilon' >0$,
$NDTIME((\ell -2)\lceil \frac 12 \log_2 n \rceil ,\ n)\nsubseteq $
\newline
$DTIME(n^{\omega  \ell /6-\epsilon'})$ holds.
\end{theorem}

The MWT problem seems harder than the DT one, it is not clear how fast
arithmetic matrix multiplication could be used here. More generally,
Abboud et al. conjectured in \cite{ABD} that the problem of determining a
minimum weight $\ell$-clique in an edge weighted graph on $N$ vertices (MWC-$\ell$)
does not admit a truly sub-$N^{\ell}$-time algorithm. The derivation
of consequences of this conjecture for the separation between
nondeterminism and determinism is quite analogous to the proof of
Theorem \ref{theo: KT}.  The differences follow from the fact that
now the bound is $N^{\ell}$ instead of $N^{\omega \ell/3}$ and the size
$n$ of the input is $N^{2+k+o(1)}$ (like in Theorem \ref{theo: APSP})
instead of $N^2\log_2 N.$ The proof details are left to the reader.

\begin{theorem}\label{theo: MT}
Let $k\ge 0,$
and let  $Q$ stand for the set of integers
\newline
$\{ -2^{N^{k+o(1)}},...,2^{N^{k+o(1)}}\}.$
Next, let $\ell \ge 3$.
If for any $\epsilon >0,$
the problem of detecting a minimum weight  $\ell$-clique
in a graph with edge weights
in $Q$ and $N$ vertices does not admit
an $O(N^{\ell - \epsilon +k+o(1)})$-time algorithm
then for any positive $\epsilon' >0$,
$NDTIME((\ell -2)\lceil \frac 1{2+k} \log_2 n \rceil,\ n)
\nsubseteq DTIME(n^{1+\frac {\ell-2}{2+k}-\epsilon'})$ holds.
\end{theorem}

\junk{
\section{Reducing the range}
Consider an instance of the problem of verifying if
an integer $N\times N$ matrix $K$ defines a metric.
Let $e_1,...,e_m$ be a sequence 
of the entries of the matrix $K$
sorted in non-decreasing order. Suppose that $e_i\ge 1$
and $3e_i< e_{i+1}$ for some $i<m.$ Then, we form
the sequence $e'_1,...,e'_i,e'_{i+1},...,e'_m,$
where for $\ell=1,...,i,$ $e'_{\ell}=e_{\ell}$
and for $\ell=i+1,....,m,$
$e'_{\ell}=e_{\ell}-e_{i+1}+3e_i$. 
We shall denote such a transformation of
the sequence  $(e_{\ell})_{\ell=1}^m$ into
the sequence $(e'_{\ell})_{\ell=1}^m$ by $Reduce((e_{\ell})_{\ell=1}^m,i)$.
Let $K'$ be the integer $N\times N$ matrix
such that $K'[u,v]=e'_{\ell}$ iff $K_[u,v]=e_{\ell}.$ 

\begin{lemma}\label{lem: reduce}
The matrix $K$ defines a metric 
if and only if the matrix $K'$ defines a metric.
\end{lemma}
\begin{proof}
\junk{Recall that 
the $N\times N$ matrix $K$ defines a metric
iff  $K[i,j]=
K[j,i]$ for all $i,\ j\in [N],$ and  $K[i,j]\le K[i,k]+K[k,j$
for all $i,\ k,\ j\in [N].$}
It is sufficient show that the symmetry and triangle
inequality conditions hold for $K'$ iff they hold for $K.$

First, let us consider the symmetry condition.
We may assume w.l.o.g. that $K'[u,v]=K[u,v]$
and  $K'[v,u]\neq K[v,u]$ or {\em vice versa}.
By symmetry, it is sufficient to consider
only the former subcase.
We have
$K'[u,v]<K'[v,u]$ by the definition of $K'.$
But then, we have also $K[u,v] < K[v,u]$
since $K'[v,u]<K[v,u].$

Consider in turn the triangle inequality condition.  We may assume
w.l.o.g. that at least one and at most two entries among
$K'[u,v],\ K'[u,w],$ $K'[w,v]$ are different from the corresponding
entries $K[u,v],\ K[u,w],$ $K[w,u]$, respectively.

Suppose first that $K'[u,v]\neq K[u,v].$ 
Then, if $K'[u,w]=K[u,v]$ and $K'[w,v]=K[w,v]$,
we have $K'[u,v]>K'[u,w]+K'[w,v]$ and
consequently $K[u,v]>K[u,w]+K[w,v]$.
Therefore, assume that $K'[u,w]\neq K[u,v]$ and $K'[w,v]=K[w,v]$.
Then, we obtain $(K'[u,w]-K'[u,v])+K'[w,v]$
$=(K[u,w]-K[u,v])+K[w,v]$
so $K'[u,v]\le K'[u,w]+K'[w,v]$ holds iff
$K[u,v]\le K[v,w]+K[w,v]$ holds. The subcase
$K'[u,w]=K[u,v]$ and $K'[w,v]\neq K[w,v]$ follows
by symmetry.

Suppose in turn that $K'[u,v]=K[u,v].$
Then, since we may assume w.l.o.g. that
$K'[u,w]\neq K[u,v]$ or $K'[w,v]\neq K[w,v]$,
we have $K'[u,v]<K'[u,w]+K'[w,v]$ as well as
$K[u,v]<K[u,w]+K[w,v]$.
\qed
\end{proof}
\begin{theorem}\label{theo: reduce}
Let $K$ be an $N\times N$ matrix with integer entries
in $\{0,1,...,M\}.$ We can construct an $N\times N$
matrix $K''$ with integer entries in $\{0,1,....,3^{N^2}\}$
in $O(N^2\log M)$ time such that $K$ defines
a metric if and only if  $K''$ does it.
\end{theorem}
\begin{proof}
Let $e_1,...,e_m$ be a sequence 
of the entries of the matrix $K$
sorted in non-decreasing order.
We run the following algorithm
that iteratively applies the Reduce
operation to the sequence.
\par
\vskip 3pt
\noindent
$(d_{\ell})_{\ell=1}^m\leftarrow (e_{\ell})_{\ell=1}^m$
\par
\noindent
{\bf for} $i=1,...,m-1$ {\bf do}
\par
\noindent
{\bf if} $d_i\ge 1$ and $d_{i+1}>3d_i$ {\bf then}
$(d_{\ell})_{\ell=1}^m\leftarrow  Reduce((d_{\ell})_{\ell=1}^m,i)$
\par
\noindent
Output $(d_{\ell})_{\ell=1}^m$
\par
\vskip 3pt
\noindent
Let $i_0$ be the smallest index $i$ such that $d_i>0$ holds in the output sequence $(d_{\ell})_{\ell=1}^m$.
Note that in the aforementioned sequence, for 
$i=i_0,...,m-1,$ $d_{i+1}\le 3d_i$ holds. Hence, since $m\le N^2,$
all elements in the sequence are in
$\{0,1,....,3^{N^2}\}.$ 

Let $K''$ be the integer $N\times N$ matrix
such that $K''[u,v]=d_{\ell}$ iff $K_[u,v]=e_{\ell}.$ 
By Lemma \ref{lem: reduce}
and induction on the iterations  of the algorithm,
we infer that
$K$ defines a metric iff $K''$ does it.

The Reduce operation can be easily implemented in $O(N^2\log M)$ time
and hence the whole algorithm can be implemented in $O(N^4\log M)$
time. For our purposes, we need however a substantially faster
implementation. The idea is simple.
When a new gap between two consecutive entries
$f,\ g$ qualifying for the next Reduce update is encountered,
we update only the value $t$ by which the entries following the gap
should be decreased. The overall time taken by this faster
implementation depicted below is easily seen to be $O(N^2 \log M).$
\junk{
For this reason, while performing the Reduce operation
for the $r$-th time, we proceed as follows. We stop the $(r-1)$-th
Reduce operation at the first edge $e$ whose weight should be updated
and only pass the value $t$ by which the weights of $e$ and the
following edges in the non-decreasing weight order should be decreased
to the $r$-th Reduce operation. The latter continues to scan the edges
in the aforementioned order starting from the edge $e$ and decreases
their weights by $t$ until a weight gap between two consecutive edges
$f,\ g$ qualifying for the next update is encountered. Then, the new
larger value of $t$ by which the weight of $g$ and the weights of
following edges should be updated is computed and passed it to the
next Reduce operation. The overall time taken by this faster
implementation depicted below is easily seen to be $O(N^2 \log M).$}
\par
\vskip 5pt
\noindent
$(d_{\ell})_{\ell=1}^m\leftarrow (e_{\ell})_{\ell=1}^m$
\par
\noindent
$t\leftarrow 0$
\par
\noindent
{\bf for} $i=1,...,m-1$ {\bf do}
\par
\noindent
$\ \ \ d_i\leftarrow d_i-t$
\par
\noindent
$\ \ \ ${\bf if} $d_i\ge 1$ and $d_{i+1}-t>3d_i$ {\bf then}
$t\leftarrow t+(d_{i+1}-t-3d_i)$
\par
\noindent
Output $(d_{\ell})_{\ell=1}^m$
\qed
\end{proof}
  By combining
  Theorem \ref{theo: reduce} with  Fact \ref{fact: VWW},
we obtain the following theorem.
\junk{
\begin{theorem}\label{theo: APSPf}
If for any $\epsilon > 0,$
the APSP problem or any of the subcubic-time equivalent
problems for graphs with integer edge weights or
integer matrix entries does not
admit a truly subcubic algorithm then
for any $\epsilon' >0,$
$NDTIME(\frac 12\log_2 n,\ n)\nsubseteq DTIME(n^{1.25-\epsilon'}).$
\end{theorem}
\begin{proof}
Let $P$ be a problem on the list that does not admit
a truly subcubic algorithm
when the edge weights or matrix entries are integers.
By Theorem \ref{theo: reduce},
we may assume w.l.o.g. that $P$ has the aforementioned property
when the edge weights or matrix entries are integers
in $\{ -3^{N^{2+o(1)}},...,3^{N^{2+o(1)}}\}.$ More exactly,
$P$ does not admit a substantially subcubic-time
algorithm when the edge weights or matrix entries
are integers in $\{ -2^{N^{k+o(1)}},...,2^{N^{k+o(1)}}\}$
for some $k\in [0,2].$ Now, when we apply
Theorem \ref{theo: APSP}, we obtain the strongest
implication when $k=0$ and the weakest when $k=2.$
Since, we do not $k,$ we  can infer only the
weakest implication.
\qed
\end{proof}}

\begin{theorem}\label{theo: APSPf}
  Let $k>2,$ $\epsilon =\min\{1,k-2\},$
and let  $Q$ stand for the set of integers
\newline 
$\{ -2^{N^{k+o(1)}},...,2^{N^{k+o(1)}}\}.$
The APSP problem or any of the subcubic-time equivalent
problems on the list from \cite{VWW}, for graphs on $N$ vertices  with edge weights in $Q$
or $N\times N$ matrices with entries in $Q$,
admits an $O(N^{3-\epsilon + k+o(1)})$-time algorithm.
\end{theorem}
\begin{proof}
  Consider the MDM problem for  $N\times N$ matrices with  entries in $Q.$
  By Theorem \ref{theo: reduce},
  it reduces in $O(N^{2+k+o(1)})$ time to the MDM problem for  $N\times N$ matrices with
  entries in $\{0,...,2^{N^{2+o(1)}}\}.$
  The latter MDM problem can be solved by a straightforward cubic algorithm in $O(N^{3+2+o(1)})$ time.
We conclude that
  the original MDM problem for $N\times N$ matrices with  entries in $Q.$
  can be solved in $O(N^{3-\epsilon +k+o(1)}$ time,
  where $\epsilon =\min\{1,k-2\}$. Consequently, by Fact \ref{fact: VWW},
  the subcubic equivalent problems on the list with edge weights
  or matrix entries in $Q$ can be also solved in $O(N^{3-\epsilon+k+o(1)})$ time.
  \qed
  \end{proof}

\begin{theorem}\label{theo: APSPf}
Let $k\ge 0,$
and let  $Q$ stand for the set of integers
\newline 
$\{ -2^{N^{k+o(1)}},...,2^{N^{k+o(1)}}\}.$
If for any $\epsilon > 0,$
the APSP problem or any of the subcubic-time equivalent
problems on the list from \cite{VWW}, for graphs on $N$ vertices  with edge weights in $Q$
or $N\times N$ matrices with entries in $Q$, does not
admit an $O(N^{3-\epsilon + k+o(1)})$-time algorithm then
for any $\epsilon' >0,$
\newline
$NDTIME(\lceil \max(\frac 14, \frac 1{2+k})\log_2 n \rceil,\ n)\nsubseteq
DTIME(n^{1+\max(\frac 14, \frac 1{2+k})-\epsilon'})$ holds.
\end{theorem}
\begin{proof}
If $k\le 2$ then the theorem immediately follows from Theorem \ref{theo: APSP}.
So, let us suppose $k>2.$
Let $P$ be a problem on the list that does not admit
a truly subcubic algorithm
when the edge weights or matrix entries are integers
in $Q.$
By Theorem \ref{theo: reduce} and Fact \ref{fact: VWW},
we may assume w.l.o.g. that $P$ has the aforementioned property
when the edge weights or matrix entries are integers
in $\{ -3^{N^{2+o(1)}},...,3^{N^{2+o(1)}}\}.$ 
Now, it is sufficient to apply Theorem \ref{theo: APSP}
with $Q=\{ -2^{N^{2+o(1)}}\,...,2^{N^{2+o(1)}}\}$.
\junk{More exactly,
$P$ does not admit a truly subcubic
algorithm when the edge weights or matrix entries
are integers in $\{ -2^{N^{\ell+o(1)}}\,...,2^{N^{\ell+o(1)}}\}$
for some $\ell\in [0,2].$ Now, when we apply
Theorem \ref{theo: APSP}, we obtain the strongest
implication when $k=0$ and the weakest when $\ell=2.$
Since, we do not $\ell,$ we  can infer only the
weakest implication.}
\qed
\end{proof}
\section{Quantum algorithms for MWT and MDM}
In this marginal section, we present simple quantum algorithms
for finding a triangle of minimum total edge length 
in an edge weighted graph (MWT) and
for verifying if a matrix defines a metric (MDM). Our quantum algorithm
for MWT
is substantially faster than the fastest  known quantum
algorithm for APSP in the general case \cite{NV} and on the other
hand substantially slower than the fastest known
algorithm for for DT 
\cite{MS}, see Table 2 (Appendix). Note that DT can be regarded as
a special case of MWT.                                                                                       

We shall
use a specialized variant
of Grover's search due to D\"{u}rr and H{\o}yer \cite{A04,DH96}
for finding an entry of the minimum value in a table.
\begin{fact}\label{fact:dur}(D\"{u}rr and H{\o}yer \cite{DH96})
Let $T[k],$ $1\le k\le n,$ be an unsorted table where
all values are distinct. Given an oracle for $T,$
the index $k$ for which $T[k]$ is minimum can be 
found by a quantum algorithm with probability
at least $\frac 12$ in $O(\sqrt n)$ time.
\end{fact}

\noindent
{\bf Algorithm 1}
\par
\noindent
{\em Input:} 
an oracle $W$ for the weighted adjacency matrix
representing an edge weighted graph on $N$ vertices
and 
a positive integer $M$ such that the edge weights
are in the range $\{-M,...,M\}.$
\par
\noindent
{\em Output:} 
a minimum weight triangle in the graph.
\par
\noindent
Set an oracle for the virtual table $T[i,j,k],$ $i,j,k\in [N],$
such that $T[i,j,k]=(N+1)^4(W[i,j]+W[i,k]+W[k,j])$ $+(N+1)^2i+(N+1)j+k$ 
if $W[i,j],\ W[i,k],\ W[k,j]$ are defined
and $T[i,j,k]=(N+1)^4(3M+1)$ $+(N+1)^2i+(N+1)j+k$ otherwise
\par
\noindent
Find the indices $i',j',k'$ minimizing $T[i,j,k]$ by using the method from Fact 
\ref{fact:dur}
\par
\noindent
{\bf if} $T[i',j',k']< (n+1)^4(3M+1)$ {\bf then} 
{\bf return} $(i',j',k')$ {\bf else} {\bf return} ``No''
\par
\vskip 4pt
\noindent
Since the values of the entries of the table $T$ in Algorithm 1 are distinct,
the method from Fact \ref{fact:dur} can be applied to $T.$ Hence, by Fact \ref{fact:dur}, we obtain
the following theorem.

\begin{theorem}
Let $G$ be a graph on $N$ vertices with integer edge weights. 
Given an oracle for the weighted adjacency matrix of $G,$
a minimum weight triangle in $G$ (if any) can be detected
by a quantum algorithm with high probability
in $\tilde{O}(N^{1.5})$ time.
\end{theorem}

In similar fashion, we obtain a quantum algorithmic solution
to the problem of verifying if a matrix defines a metric.

\begin{theorem} 
  Given an oracle for an $N\times N$ integer matrix,
  the problem of verifying if the matrix
  defines a metric admits an $\tilde{O}(N^{1.5})$-time quantum algorithm.
  \end{theorem}
\begin{proof}
  The algorithm is similar to that
  quantum for MWT, i.e., Algorithm 1.
  For each of the four properties that the matrix, say $K,$ should
  have, we form a virtual table with distinct values.
  On the base of the minimum value of an entry in the table,
  we can decide if the property holds. For searching for
  the minimum, we use again Fact \ref{fact:dur}.
  The largest of the virtual tables is of cubic in $N$ size
  and it corresponds to the triangle inequality property.
  For $i,j,k\in [N],$ the virtual table is defined by
  $T[i,j,k]=(N+1)^4(K[i,k]+K[k,j]-K[i,j])$ $+(N+1)^2i+(N+1)j+k$.
  Note that all values of the entries of $T$ are distinct
  and that the triangle inequality is violated by $K$ if
  and only if the minimum value of an entry of $T$ is
  negative. The virtual quadratic table $U$ corresponding to the
  symmetry property is given by
  $U[i,j]=(N+1)^4(U[i,j]-U[j,i])^2$ $+(N+1)^2i+(N+1)j+k$ for
  $i,j \in [N].$ We leave the rest of details to the reader.
  The application of the quantum search from Fact \ref{fact:dur}
  to the verification of the triangle inequality property dominates
  the time complexity.
  \qed
\end{proof}

\begin{table*}[t]
\begin{center}
\begin{tabular}{||c|c|c||} \hline \hline
problem & ~upper bound~ & author
\\ \hline \hline
APSP  &  $\tilde{O}(N^{2.5})$
& Navebi and Vassilevska Williams  \cite{NV}
\\ \hline
MWT & $\tilde{O}(N^{1.5})$ & This paper
\\ \hline
MDM  & $\tilde{O}(N^{1.5})$ & This paper
	\\ \hline
DT  & $\tilde{O}(N^{9/7})$ & Lee, Magniez, and Santha \cite{MS}
\\ \hline
3SUM & $\tilde{O}(N^{1+o(1)})$ & Ambainis \cite{A04}
\\ \hline
3SAT & $1.153^Npoly(N)$ & Ambainis \cite{A04}
\\ \hline \hline
\end{tabular}
\label{table: 1}
\vskip 0.5cm
\caption{Known upper bounds on the time complexity
  of quantum algorithms for problems discussed in this paper.
  In case of 3SAT $N$ stands for the number of variables
  in the input formula. }
\end{center}
\end{table*}
}
\section{Implication of 3SUM hardness}

Gajentaan and Overmars \cite{GO} exhibited a large class of geometric problems
that were the so called 3SUM hard, i.e., if any of them admitted
a substantially subquadratic algorithm then 3SUM would also have a substantially subquadratic algorithm.
The aforementioned class has been subsequently expanded by not necessarily geometric
problems (e.g., \cite{DKP,LPV,VWW}). In this section,
we show that if 3SUM, for integers within a bounded (up to super-exponential) range,
does not admit a truly subquadratic algorithm then a strong separation between deterministic
and non-deterministic time holds. The proofs in this section are similar to those
from Section 3. The key trick in the proof of the following lemma is analogous to that
in the footnote on 3SUM on page 2 in \cite{Ryan}.

\begin{lemma}\label{lem: guess1}
The 3SUM problem, when the input $N$ numbers
are integers in $\{ -2^{N^{k+o(1)}},...,2^{N^{k+o(1)}}\}$ for some $k\ge 0,$
is  in $NDTIME(\lceil \log_2 N\rceil , N^{1+k+o(1)})$.
\end{lemma}
\begin{proof}
In order to guess a number $q$ that belongs to a triple of numbers
whose sum is zero, a non-deterministic $\lceil \log_2 N \rceil$-bit
oracle is sufficient.  A multi-tape Turing machine can verify
if the guessed number $q$ belongs to such a triple as follows.
First, it sorts the input numbers in non-decreasing order
in $N^{1+k+o(1)}$ time. Next, it places two copies of
the sorted sequence on two tapes and moves its heads over the tapes
in opposite directions starting from the opposite ends.
When its head over the first tape advances to a number $r$
then its head over the second tape advances in opposite
direction to check if the number $-q-r$ occurs in the sorted
sequence, using auxiliary tapes.
The whole process takes $N^{1+k+o(1)}$ time.
In this way, the Turing machine
can verify if $q$ belongs to a triple of numbers summing to zero
in $N^{1+k+o(1)}$ time.
\qed
\end{proof}

The proof of the following theorem is analogous to that
of Theorem \ref{theo: APSP} in Section 3.

\begin{theorem}\label{theo: 3SUM}
 If for any $\epsilon>0,$ the 3SUM problem, when the input $N$ numbers
 are integers in $\{ -2^{N^{k+o(1)}},...,2^{N^{k+o(1)}}\},$ for some $k\ge 0,$
 does not admit 
 an $O(N^{2-\epsilon + k+o(1)})$-time algorithm then
 for any $\epsilon' >0,$
 \newline
 $NDTIME(\lceil \frac 1{1+k}\log_2 n \rceil, n) \nsubseteq $
 $DTIME(n^{1+\frac 1{1+k}-\epsilon'})$ holds.
\end{theorem}
\begin{proof}
By the theorem assumptions, we infer
that for any $\epsilon >0$,
the 3SUM problem, 
when the input numbers are in the range $\{-2^{-N^{k+o(1)},...,2^{-N^{k+o(1)}}\}},$
does not admit an $O(N^{2-\epsilon+k +o(1)})$-time
RAM algorithm under the logarithmic cost.
On the other hand,
the so restricted
3SUM problem is in
\newline
$NDTIME(\lceil \log_2 N\rceil , N^{1+k+o(1)})$
by Lemma \ref{lem: guess1}. 
Consequently, if we assume $N^{1+k+o(1)}$-bit representation
of the input numbers and
set $n=N^{1+k+o(1)}$ then we can conclude that 
the restricted 3SUM is
in
$NDTIME(\lceil \frac 1{1+k}\log_2 n \rceil, n).$  Now the proof is by contradiction.
Suppose that 
$NDTIME(\lceil \frac 1{1+k}\log_2 n \rceil, n)\subseteq$
$DTIME(n^{1+\frac 1{1+k}-\epsilon'})$ holds.
Then, the restricted 3SUM problem admits an 
$O((N^{1+k+o(1)})^{1+\frac 1{1+k}-\epsilon'})$-time
algorithm, i.e., an
$O(N^{2+k-(1+k+o(1))\epsilon' +o(1)})$-time algorithm in
the multi-tape (deterministic) Turing machine model.
A multi-tape (deterministic) Turing machine of time complexity
$T(n)\ge n$ can be easily simulated by a RAM with logarithmic cost
running in $O(T(n)\log T(n))$ time \cite{AHU74}.
We conclude that the restricted 3SUM problem admits
an
\newline
$O(N^{2+k-(1+k+o(1))\epsilon' +o(1)}\log N )$-time, i.e.,
an
$O(N^{2+k-(1+k+o(1))\epsilon' +o(1)})$-time algorithm, in the RAM model.
We obtain a contradiction.
\qed
\end{proof}

In fact, 3SUM is known to be hard,
under randomized reductions, when $k=0,$
more precisely, already when the numbers are in
$\{-n^3,…,n^3\}$ \cite{ALW}. The “strong 3SUM conjecture” even states that it
is hard when the numbers are in $\{-n^2,…,n^2\}.$

\section{Extension to ETH and SETH}

We can even easily derive implications of similar form from
the Exponential Time Hypothesis (ETH), involving a logarithmic
non-deterministic oracle. Let $kSAT(\ell)$ denote the kSAT
problem (see \cite{AHU74})
for instances with $\ell$ variables and distinct clauses.
Roughly, ETH conjectures that $3SAT(\ell)$ requires $2^{\Omega(\ell)}$
(deterministic) time  \cite{IP} while its strong version SETH  \cite{CIP} conjectures
that when $k$ tends to infinity than the time complexity
of $kSAT(\ell)$ tends to at least $2^{\ell}$ 

\begin{lemma}\label{lem: eth}
  $3SAT(\ell)$  
  is in $NDTIME(\ell, poly(\ell)).$
\end{lemma}
  \begin{proof}
  In order to guess an assignment (if any) satisfying
  the input $3SAT(\ell)$ formula,
  a non-deterministic $\ell$-bit oracle is enough.
  Now it is sufficient to observe that w.l.o.g. the number of clauses
  in the input formula is at most $8\binom {\ell} 3.$
  \qed
\end{proof}

\begin{theorem}\label{theo: eth}
If ETH holds then there is $k>0$ such that
\newline
$NDTIME(\lceil k\log_2 n \rceil , n) \nsubseteq DTIME(n).$
\end{theorem}
  \begin{proof}
  By ETH, there is a constant $k_0$ such that
  $3SAT(\ell)$ does not admit any $2^{k_0\ell}poly(\ell)$-time algorithm.
  Let $T(\ell)$ be the time taken by the non-deterministic algorithm for $3SAT(\ell)$
  from Lemma \ref{lem: eth}. 
  We may assume w.l.o.g. that $\ell$ is enough large so
  the inequality $T(\ell)\le 2^{k_0\ell}$ holds. Hence, it follows from
  Lemma \ref{lem: eth} that $3SAT(\ell)$ is in $NDTIME(\ell , 2^{k_0\ell}).$ Let us set
  $n=2^{k_0\ell}.$  Then, appropriately padded
  $3SAT(\ell)$ is in $NDTIME(\lceil k_1\log_2 n \rceil , n)$
  for $k_1=\frac 1 {k_0}$
  and the containment $NDTIME(k_1\log_2 n , n) \subseteq DTIME(n)$
  cannot hold since it would contradict the non-existence of
  $2^{k_0\ell}poly(\ell)$-time algorithm for $3SAT(\ell).$
  \qed
  \end{proof}
\junk{  
By assuming the strong ETH (SEHT) \cite{CI}
instead of ETH the implication of
Theorem \ref{theo: eth} can be easily strengthened
so for a given $\delta >1,$
the coefficient $k$ at $\log_2 n$ is in $(1,\delta).$
}
The Orthogonal Vectors problem in dimension
$d$ ($OV(d)$) is for two sets $A$, $B$ of  vectors
in $\{ 0,1\}^d$ to detect a pair of vectors $a\in A$ and $b\in B$
such that $a,\ b$ are orthogonal in $Z^d.$

\begin{lemma}\label{lem: ov}
  For any natural number $d,$ the $OV(d)$ problem is in
  \newline
  $NDTIME(\lceil \log N \rceil, Nd).$
\end{lemma}
  \begin{proof}
   To guess a vector $a\in A$ which is orthogonal
  to some vector $b\in B$, a $\lceil \log N \rceil$-bit
  non-deterministic oracle is sufficient. It remains to compute
  the inner product of $a$ with each vector in $B.$  This
  can be easily accomplished by a multi-tape Turing
  machine in $O(Nd)$ time.
  \qed
  \end{proof}

The low-dimension OV conjecture asserts that the $OV(d)$ problem
does not admit a truly sub-$N^2$-time algorithm when
$d=\omega (\log N)$ \cite{GIK}. It is implied by SETH \cite{Ryan1}.

\begin{theorem}\label{theo: seth}
  If the low-dimension OV conjecture holds, and hence if
  SETH holds, then for any $\epsilon>0,$
  $NDTIME(\lceil \log_2 n \rceil, n)\nsubseteq
  DTIME(n^{2-\epsilon})$ holds.
\end{theorem}
\begin{proof}
Consider the OV problem, where $d=\omega (\log N)$ and on the other
hand, $d=N^{o(1)}.$ By the theorem assumption, it does not
admit a truly sub-$N^{2}$-time algorithm. Suppose that for some
$\epsilon >0,$ $NDTIME(\lceil \log_2 n \rceil, n)\subseteq
DTIME(n^{2-\epsilon})$ holds. Set $n=Nd.$ It follows from Lemma \ref{lem: ov}
that the OV problem is in $NDTIME(\lceil \log_2 N \rceil, Nd)\subset
NDTIME(\lceil \log_2 n \rceil, n)$. Hence, it can be solved
by a Turing machine operating in $O(n^{2-\epsilon})$ time,
and consequently by a RAM with logarithmic cost in
$O((N^{1+o(1)})^{2-\epsilon})\log N)=O(N^{2-\epsilon +o(1)})$ time.
We obtain a contradiction.
\qed
\end{proof}

\section{Quantum algorithms for MWT and MDM}
In this marginal section, we present simple quantum algorithms
for finding a triangle of minimum total edge length 
in an edge weighted graph (MWT) and
for verifying if a matrix defines a metric (MDM). Our quantum algorithm
for MWT
is substantially faster than the fastest  known quantum
algorithm for APSP in the general case \cite{NV}. On the other
hand, it is substantially slower than the fastest known
quantum algorithm for DT 
\cite{MS}, see Table 2 (Appendix). Note that DT can be regarded as
a special case of MWT.                                                                                       

We shall
use a specialized variant
of Grover's search due to D\"{u}rr and H{\o}yer \cite{A04,DH96}
for finding an entry of the minimum value in a table.
\begin{fact}\label{fact:dur}(D\"{u}rr and H{\o}yer \cite{DH96})
Let $T[k],$ $1\le k\le n,$ be an unsorted table where
all values are distinct. Given an oracle for $T,$
the index $k$ for which $T[k]$ is minimum can be 
found by a quantum algorithm with probability
at least $\frac 12$ in $O(\sqrt n)$ time.
\end{fact}

\noindent
{\bf Algorithm Q}
\par
\noindent
{\em Input:} 
an oracle $W$ for the weighted adjacency matrix
representing an edge weighted graph on $N$ vertices
and 
a positive integer $M$ such that the edge weights
are in the range $\{-M,...,M\}.$
\par
\noindent
{\em Output:} 
a minimum weight triangle in the graph.
\par
\noindent
Set an oracle for the virtual table $T[i,j,k],$ $i,j,k\in [N],$
such that $T[i,j,k]=(N+1)^4(W[i,j]+W[i,k]+W[k,j])$ $+(N+1)^2i+(N+1)j+k$ 
if $W[i,j],$ $W[i,k],$ $W[k,j]$ are defined
and $T[i,j,k]=(N+1)^4(3M+1)$ $+(N+1)^2i+(N+1)j+k$ otherwise
\par
\noindent
Find the indices $i',j',k'$ minimizing $T[i,j,k]$ by using the method from Fact 
\ref{fact:dur}
\par
\noindent
{\bf if} $T[i',j',k']< (n+1)^4(3M+1)$ {\bf then} 
{\bf return} $(i',j',k')$ {\bf else} {\bf return} ``No''
\par
\vskip 4pt
\noindent
Since the values of the entries of the table $T$ in Algorithm Q are distinct,
the method from Fact \ref{fact:dur}can be applied to $T.$ Hence, by Fact \ref{fact:dur}, we obtain
the following theorem.

\begin{theorem}
Let $G$ be a graph on $N$ vertices with integer edge weights. 
Given an oracle for the weighted adjacency matrix of $G,$
a minimum weight triangle in $G$ (if any) can be detected
by a quantum algorithm with high probability
in $\tilde{O}(N^{1.5})$ time.
\end{theorem}

In similar fashion, we obtain a quantum algorithmic solution
to the problem of verifying if a matrix defines a metric.

\begin{theorem} 
  Given an oracle for an $N\times N$ integer matrix,
  the problem of verifying if the matrix
  defines a metric admits an $\tilde{O}(N^{1.5})$-time quantum algorithm.
  \end{theorem}
\begin{proof}
  The algorithm is similar to that
  quantum for MWT, i.e., Algorithm Q.
  For each of the four properties that the matrix, say $K,$ should
  have, we form a virtual table with distinct values.
  On the base of the minimum value of an entry in the table,
  we can decide if the property holds. For searching for
  the minimum, we use again Fact \ref{fact:dur}.
  The largest of the virtual tables is of cubic in $N$ size
  and it corresponds to the triangle inequality property.
  For $i,j,k\in [N],$ the virtual table is defined by
  $T[i,j,k]=(N+1)^4(K[i,k]+K[k,j]-K[i,j])$ $+(N+1)^2i+(N+1)j+k$.
  Note that all values of the entries of $T$ are distinct
  and that the triangle inequality is violated by $K$ if
  and only if the minimum value of an entry of $T$ is
  negative. The virtual quadratic table $U$ corresponding to the
  symmetry property is given by
  $U[i,j]=(N+1)^4(U[i,j]-U[j,i])^2$ $+(N+1)^2i+(N+1)j+k$ for
  $i,j \in [N].$ We leave the rest of details to the reader.
  The application of the quantum search from Fact \ref{fact:dur}
  to the verification of the triangle inequality property dominates
  the time complexity.
  \qed
\end{proof}

\begin{table*}[t]
\begin{center}
\begin{tabular}{||c|c|c||} \hline \hline
problem & ~upper bound~ & author
\\ \hline \hline
APSP  &  $\tilde{O}(N^{2.5})$
& Navebi and Vassilevska Williams  \cite{NV}
\\ \hline
MWT & $\tilde{O}(N^{1.5})$ & This paper
\\ \hline
MDM  & $\tilde{O}(N^{1.5})$ & This paper
	\\ \hline
DT  & $\tilde{O}(N^{9/7})$ & Lee, Magniez, and Santha \cite{MS}
\\ \hline
3SUM & $\tilde{O}(N^{1+o(1)})$ & Ambainis \cite{A04}
\\ \hline
3SAT & $1.153^Npoly(N)$ & Ambainis \cite{A04}
\\ \hline \hline
\end{tabular}
\label{table: 2}
\vskip 0.5cm
\caption{Known upper bounds on the time complexity
  of quantum algorithms for problems discussed in this paper.
  In case of 3SAT $N$ stands for the number of variables
  in the input formula. }
\end{center}
\end{table*}
\section*{Final remark}
Among the implications derived from the conjectures considered
in this paper, the strongest seem to be those from the conjectures on 3SUM,
MWC-$\ell$, and SETH.
The weakest seems the implication from ETH.
See Table 1.
  \section{Acknowledgments}
  The author is very grateful
  to reviewers of an earlier version of this
  paper for valuable comments and suggestions.
\junk{
We can immediately generalize our results by assuming the
non-existence of an $O(N^{\lambda-\epsilon + k+o(1)})$-time algorithm
instead of an $O(N^{3-\epsilon + k+o(1)})$-time algorithm in Theorems
\ref{theo: APSP}, \ref{theo: APSPf} or instead of an $O(N^{2-\epsilon +
  k+o(1)})$-time algorithm in Theorem \ref{theo: 3SUM}, respectively.  Then,
in the implied negated inclusions the term $\frac 1{2+k}$ in the
exponent is replaced by $\frac {\lambda -2}{2+k}$ in Theorem
\ref{theo: APSPf}, the term $\max ( \frac 14,\frac 1{2+k})$
is replaced by $(\lambda
-2)\max (\frac 14,\frac 1{2+k})$ in Theorem \ref{theo: APSPf}, and
the term $\frac 1{1+k}$ is replaced by $\frac {\lambda -1}{1+k}$ in Theorem \ref{theo:
  3SUM}, respectively. In order to obtain a similar generalization of
Theorem \ref{theo: DT}, it is sufficient to replace $\omega $ by
$\lambda \le \omega$, appropriately.

The $P=NP$ question is equivalent to the question if there exist $k\ge
1$ such that $NDTIME(O(n),n)\subseteq DTIME(n^k).$ In turn, the latter
question can be extended by a standard padding argument to include the
question if there exist $k\ge 1$ such that $NDTIME(n^{1/k},n)\subseteq
DTIME(n).$ Of course, $P\neq NP$ implies the negative answer to the
extended question again by a padding argument but the reverse
implication does not seem to hold necessarily.  Hence, getting the
negative answer to the extended question might be easier than
disproving $P=NP.$ For instance, it is sufficient to show that there
is no deterministic linear-time algorithm for detecting a clique of
$N^{o(1)}$ size for this purpose.}
\junk{
In the realm of problems solvable in moderate polynomial time $T(n)$
(e.g., cubic or quadratic), one tends to use reductions in truly
sub-$T(n)$ time \cite{LPV,VWW} from other well known problems (e.g.,
APSP or 3SUM) or the exponential time hypothesis (ETH) in order to
express the relative hardness of substantially
improving the $T(n)$ bound. In this
paper, we provide an additional tool for expressing the hardness of
proving a close lower bound on $T(n)$ by showing that the latter in
case of some problems would imply an enormous breakthrough separation
result between deterministic and non-deterministic time.}

\vfill
\end{document}